\documentclass[twocolumn,3p]{elsarticle}  

\usepackage{fancyhdr}
 \usepackage{amsmath}
\usepackage{cases}
\usepackage{stfloats}
\usepackage{enumerate}
\usepackage{url}
\usepackage{multicol}
\usepackage{graphicx}
\usepackage{wrapfig}
\usepackage{picinpar}
\usepackage{cutwin}
\usepackage{picins}
\usepackage{balance}
\usepackage{algorithm}
\usepackage{algpseudocode}

 \usepackage{etoolbox}

\newtoggle{ACM}
\toggletrue{ACM}
\togglefalse{ACM}

\newtoggle{IEEEcls}
\toggletrue{IEEEcls}
\togglefalse{IEEEcls}

\newtoggle{ElsJ}
\toggletrue{ElsJ}

\floatname{algorithm}{Algorithm}

\iftoggle{IEEEcls}{
\usepackage{amsthm}
\theoremstyle{plain}
\newtheorem{theorem}{Theorem}
\newtheorem{corollary}{Corollary}
\theoremstyle{definition}
\newtheorem{definition}{Definition}
\newtheorem{property}{Proposition} 
\newtheorem{lemma}{Lemma}

}

\iftoggle{ElsJ}{
\usepackage{amsthm}
\theoremstyle{plain}
\newtheorem{theorem}{Theorem}

\theoremstyle{definition}
\newtheorem{definition}{Definition}
 
\newtheorem{lemma}{Lemma}

}

\newif\ifNotUse  
 \NotUsetrue
 
 \newif\ifNotUse  
 \NotUsetrue

 \iftoggle{ACM}{
\setcopyright{none}  
\acmDOI{}
\acmPrice{}
\acmISBN{}
}
 
\begin{document}

\title{SRLA: A real time sliding time window super point cardinality estimation algorithm for high speed network based on GPU}
\author[seu_cs]{Jie Xu\corref{cor1}}
\ead{xujieip@163.com}
\author [seu_ns] {Wei Ding}
\ead{wding@carnation.njnet.edu.cn}
\author [seu_ns] {Jian Gong}
\ead{jgong@carnation.njnet.edu.cn}
\author [seu_ns] {Xiaoyan Hu}
\ead{xyhu@carnation.njnet.edu.cn}

\cortext[cor1]{Corresponding author}
\address[seu_cs]{School of Computer Science and Engineering, South East University, Nanjing, China}
\address[seu_ns]{School of Cyber Science and Engineering, South East University, Nanjing, China}

\begin{abstract}
Super point is a special host in network which communicates with lots of other hosts in a certain time period. The number of hosts contacting with a super point is called as its cardinality. Cardinality estimating plays important roles in network management and security. All of existing works focus on how to estimate super point's cardinality under discrete time window. But discrete time window causes great delay and the accuracy of estimating result is subject to the starting of the window. sliding time window, moving forwarding a small slice every time, offers a more accuracy and timely scale to monitor super point's cardinality. On the other hand, super point's cardinality estimating under sliding time window is more difficult because it requires an algorithm can record the cardinality incrementally and report them immediately at the end of the sliding duration. This paper firstly solves this problem by devising a sliding time window available algorithm SRLA. SRLA consists of two cardinality estimating algorithms, sliding rough estimator SRE and sliding linear estimator SLE. SRE is used to detect super point while scanning packets and it generates a candidate super point list at the end of a sliding time window. With this candidate super point list, SLE estimates the cardinality of every candidate super point fast and accurately. SRLA's ability of working under sliding time window comes from a novel cardinality recorder, distance recorder DR. DR records the time a host appearing and helps SRE and SLE to judge if a host contacting with a super point is in a certain time period. SRLA could run parallel to deal with high speed network in line speed. This paper also gives the way to deploy SRLA on a common GPU. Experiments on real world traffics which have 40 GB/s bandwidth show that SRLA estimates super point's cardinality within 100 milliseconds under sliding time window when running on a Nvidia GPU GTX650 with 1 GB memory. The estimating time of SRLA is much smaller than that of other algorithms which consumes more than 2000 milliseconds under discrete time window. 
\end{abstract} 

\maketitle


\section{Introduction}

Super point cardinality estimation has been researched for a long time because of its importance\cite{JNCA2015_BotFlexACommunityDrivenToolForBotnetDetection}\cite{JNCA2016:BotnetDetectionViaMiningTrafficFlowCharacteristics}\cite{JNCA2017:AValidationModelForNonlexicalRoutingProtocols}. And many excellent algorithms have been proposed recent years\cite{HSD:AbitmapBasedAlgorithmDetectingStealthySuperpoints}\cite{HSD:2014ANewSketchMethodMeasuringHostConnectionDegreeDistribution} . But these algorithms only work for discrete time window, under which there is no duplicating time period between two adjacent windows. These algorithms will reinitialize at the beginning of every window and discard hosts' cardinality information of previous time\cite{HSD:identifyHighCardinalityHosts}. The discrete time window splits host cardinality into discrete pieces and doesn't report super point cardinality until the end of a window which has a latency of the size of time window. Sliding time window which moves a small unit smoothly has a better measurement result than discrete time window. It stores and updates host cardinality information incrementally. Sliding time window estimates super point cardinality more precisely because it is not affected by the starting of window. And sliding time window reports super point more timely for the sake that the moving step is much smaller than the size of discrete time window and at the end of each moving step, super point cardinality will be estimated immediately. But super point detection and cardinality estimation under sliding time window is more complex than that under discrete time window because it maintains hosts state of some previous time and estimates super point's cardinality more frequently.

Super point's cardinality estimation could be divided into three procedures: packets scanning, super point detection and cardinality estimation. The first procedure scans every packet and records necessary information about hosts in different networks. Procedures 2 and 3 detect super points and estimate their cardinalities according to the recording information. They always run together at the end of a time window which is what all of existing algorithms do\cite{HSD:streamingAlgorithmFastDetectionSuperspreaders}\cite{HSD:findFrequentItemsInStream}. But the merging of super point detection and cardinality estimation will consume lots of time at the end of a time window because restoring super points from huge hosts is a complex procedure. Under discrete time window, the end window procedure time will not cause many influence because two adjacent discrete time windows have no duplicate time period and duration between their end points is equal to the length of them. But under sliding time window, the duration of two windows' end points is only a small part of the window size and super point's cardinality will be estimated more frequently than that under discrete time window. In order to estimate super point's cardinality under sliding time window in real time, the end window procedure time must be much smaller than the sliding step. If we try to detect super point while scanning packets in the time window and only estimate the cardinalities of these detected candidate super points, the procedure time at the end of the time window will be reduced greatly.  

Millions of packets passing through a high speed network every second\cite{JNCA2018:MATEMAUnifiedFrameworkTustMCDMAssuringSecurityReliabilityQosDTNRouting}\cite{JNCA2018:ZeroqueueEthernetCongestionControlProtocolBasedOnAvailableBandwidthEstimation}. So the super point detection algorithm running with packets scanning must be light weight: small memory requirement, fast processing speed. This light weight detection algorithm generates a candidate super points list while scanning packets. At the end of a time window, a more accurate algorithm will be used to estimate every candidate super point's cardinality fast. This paper devises two estimators: sliding rough estimator $SRE$ and sliding long estimator $SLE$. $SRE$ is a light weigh estimator which judges if a host is super point while scanning packets. And $SLE$ calculates the cardinality of a given host. Based on these two estimators, a novel sliding super point's cardinality estimating algorithm $SRLA$ is proposed. In order to work under sliding time window, $SRLA$ uses a new method called as distance recorder $DR$ to record the appearance of a host. $DR$ helps $SRLA$ to judge if a host appears in a certain sliding time window by updating itself incrementally.  

Nowadays network bandwidth is becoming higher and higher\cite{JNCA2018:SmartResourceAllocationForImprovingQoEIPMultimediaSubsystems}\cite{JNCA2016:SoftwareDefinedNetworksASurvey}. To estimate super point's cardinality from the high speed network in real time, parallel processing technology is necessary. Most of the previous algorithms tried to accelerate the packets processing speed by used fast memory SRAM. But the small size SRAM limits the accuracy of these algorithms in a high-speed network. What's more, estimation algorithm requires lots of computation operations and the computation ability of CPU is also the bottleneck. Parallel computation ability of GPU (Graphic Processing Unit) is stronger than that of CPU because of its plenty operating cores. When using GPU to scan packets parallel, a high throughput will be acquired.

Motivated by these ideas, this paper firstly proposed a sliding time window available super point cardinality estimation algorithm $SRLA$. The main contribution of this paper is listed below.

\begin{enumerate}
\item Devise a novel light weight method to judge if a host is a super point under sliding time window.

\item Firstly propose a super point detection and cardinality estimation algorithm under sliding time window. 

\item Deploy the sliding super point detection and cardinality estimation algorithm on a common GPU to deal with core network in real time.

\end{enumerate}

In the next section, we will introduce previous super point detection algorithm under discrete time window and analyze their merit and weakness. In section 3, two sliding time window available cardinality estimators, sliding rough estimator and sliding linear estimator, are proposed. Section 4 introduces the novel algorithm $SRLA$ and describes how it does to detect super points and estimate the cardinality under sliding time window. In this section, a method to deploy $SRLA$ on GPU is also proposed. Section 5 shows experiments of real world 40Gb/s core network traffic. And we make a conclusion in the last section.

\section{Related work}
Super point detection is a hot topic in network research field. Shobha et al.\cite{HSD:streamingAlgorithmFastDetectionSuperspreaders} proposed an algorithm that did not keep the state of every host so this algorithm can scale very well. Cao et al.\cite{HSD:identifyHighCardinalityHosts} used a pair-based sampling method to eliminate the majority of low opposite number hosts and reserved more resource to estimate the opposite number of the resting hosts. Estan et al.\cite{HSD:bitmapCountingActiveFlowsHighSpeedLinks} proposed two bits map algorithms based on sampling flows. Several hosts could share a bit of this map to reduce memory consumption. All of these methods were based on sampling flows which limited its accuracy. 

Wang et al.\cite{HSD:ADataStreamingMethodMonitorHostConnectionDegreeHighSpeed} devised a novel structure, called double connection degree sketch (DCDS), to store and estimate different hosts cardinalities. They updated DCDS by setting several bits simply. In order to restore super points at the end of a time period, which bits to be updated were determined by Chinese Remainder Theory(CRT) when scanning packets. By using CRT, every bit of DCDS could be shared by different hosts. But the computing process of CRT was very complex which limited the speed of this algorithm.

Liu et al.\cite{HSD:DetectionSuperpointsVectorBloomFilter} proposed a simple method to restore super hosts basing on bloom filter. They called this algorithm as Vector Bloom Filter(VBF). VBF used the bits extracted from a IP address to decide which bits to be updated when scanning packets. Compared with CRT, bit extraction only needed a small operation.  But VBF would consume much time to restore super point when the number of super points was very big because it used only four bit arrays to record cardinalities.

Most of the previous works only focused on accelerating speed by adopting fast memory but they neglected the calculation ability of processors. Seon-Ho et al.\cite{HSD:GPU:2014:AGrandSpreadEstimatorUsingGPU} first used GPU to estimate hosts opposite numbers. They devised a Collision-tolerant hash table to filter flows from origin traffic and used a bitmap data structure to record and estimate hosts' opposite numbers. But this method needed to store IP address of every flow while scanning traffic because they could not restore super points from the bitmap directly. Additional candidate IP address storing space increased the memory requirement of this algorithm.

All of these algorithms can't not work under sliding time window because they must reinitialize their data structures at the beginning of every window. And they consume too much time when estimating cardinalities which is longer than the step of window sliding. In the following part, an incrementally updating and fast estimating algorithm is proposed to detect super point and estimate their cardinalities under sliding time window.
\section{Sliding cardinality estimation}
There are huge hosts in a high speed network. But super point takes up a little proportion. Estimating super point's cardinality under sliding time window contains two parts: detecting super points under sliding time window, estimating their cardinalities. Both of these parts have a same issue, how to estimate a host's cardinality under sliding time window. Two sliding estimators, sliding rough estimator (SRE) and sliding linear estimator (SLE), are introduced in this section for super point detection and cardinality estimation separately. SRE judges if a host is super point with small memory and SLE gives the cardinality estimation for a certain host. For verbal clarity, we firstly give the definition of sliding time window.

\subsection{sliding time window definition}
Suppose there are two networks $A$ and $B$. These two networks are contacting with each other through an edge router $ER$. $A$ might be a city-wide network or even a country-wide network. And $B$ might be another city-wide network or the Internet. All traffic between $A$ and $B$ could be observed from $ER$. Split this traffic by successive time slices as shown in figure \ref{SlidingWindowModal}.

\begin{figure}[!ht]
\centering
\includegraphics[width=0.47\textwidth]{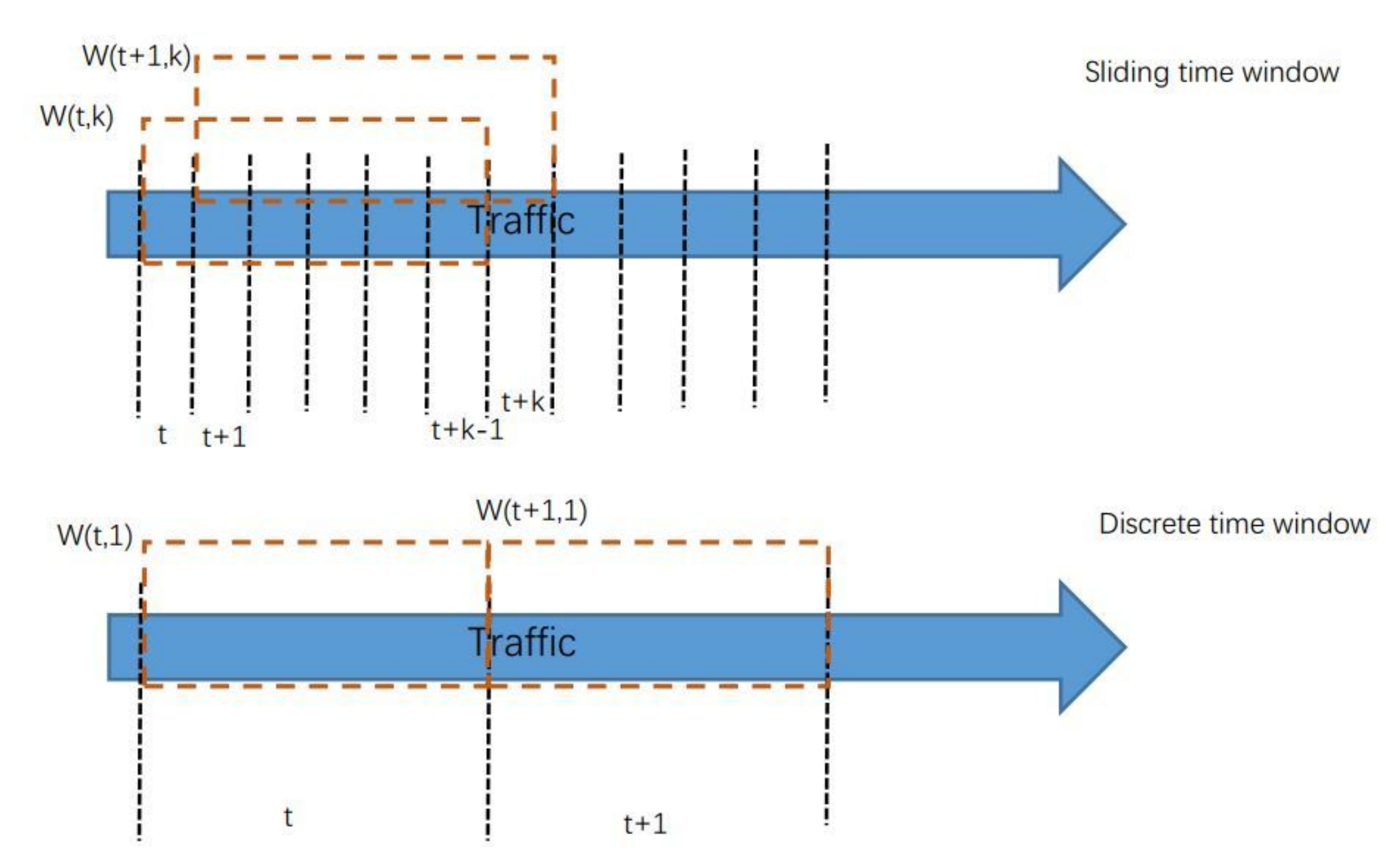}
\caption{Sliding time window and discrete time window}
\label{SlidingWindowModal}
\end{figure}

These time slices have the same duration. The length of a time slice could be 1 second, 1 minute or any period in different situations. Every time slice is identified by a number. A sliding time window $W(t,k)$ contains $k$ successive slices starting from the $t$ time slice as shown in the top part of figure \ref{SlidingWindowModal}. Sliding time window will move forward one slice once a time. So two adjacent sliding time windows contain $k-1$ same slices. When $k$ is set to 1, there is no duplicate time period between two adjacent windows, which is the case of discrete time window in the bottom part of figure \ref{SlidingWindowModal}. 

Let $A$ be the network from which we want to detect super points. A host's packets stream in a sliding time window is defined as below.

\begin{definition}[Packets stream of a host]
\label{def-slidingPktStream}
For a host $aip \in A$, every packet passing through $ER$ in sliding time window $W(t,k)$ which has $aip$ as source or destination address composes packets stream of $aip$, written as $Pkt(aip,t,k)$.
\end{definition}

$aip$'s opposite hosts stream $ST(aip,t,k)$ could be derived from $Pkt(aip,t,k)$ by extracting the other IP address except $aip$. A IP address $bip$ may appear several times in $ST(aip,t,k)$ because $aip$ can send several packets to $bip$ or receive many packets from $bip$. Hosts in $ST(aip,t,k)$ make up of opposite hosts set of $aip$, written as $OP(aip,t,k)$. The number of element in $OP(aip,t,k)$, denoted as $|OP(aip,t,k)|$, is no bigger than that of $ST(aip,t,k)$. $|OP(aip,t,k)|$ is the cardinality of $aip$ in sliding time window $W(t,k)$. Sliding super point is defined according to host's cardinality.

\begin{definition}[Sliding super point]
\label{def-slidingSuperPoint}
For a host $aip \in A$, if $|OP(aip,t,k)| \geq \theta$, $aip$ is a sliding super point in sliding time window $W(t,k)$. Where $\theta$ is a positive integer.
\end{definition}

Threshold $\theta$ is defined by users for different applications. It could be selected according to the average cardinality of all host in the past or the normal cardinality of a server. How to get $|OP(aip,t,k)|$ from $ST(aip,t,k)$ is a hard task. Because packets pass through $ER$ with high speed and every packet could only be scanned a time in the stream. How to process every coming packet, judge if it is a super point and give an accurate estimation of $|OP(aip,t,k)|$ at the end of the last time slice of $W(t,k)$ is the key step in the whole algorithm.

\subsection{Sliding rough estimator}

Cardinality estimator scans IP pair stream in a window and gives hosts' cardinalities estimation at the end of this window. A IP pair is extracted form a packet passing through $R$.
\begin{definition}{IPpair}
\label{def_IPpair}
A IP pair is a tuple of two IP addresses extracted from a packet like $<aip_i,bip_j>$ where $aip_i \in A$ and $ bip_j \in B$. IP pair stream in $W(t,k)$ is the stream of IP pairs extracted from every packet passing through $R$ in $W(t,k)$ and it is denoted by $IPpair(A,t,k)$.
\end{definition}

For a host $aip$ in $A$, its IP pair stream $IPpair(aip,t,k)$ is the sub stream of IP pairs which have $aip$ as the first IP addresses. Let $OP(aip,t,k)$ represent the set of the second IP addresses of IP pairs in $IPpair(aip,t,k)$. The task of estimate $aip$'s cardinality is to get the number of hosts in $OP(aip,t,k)$, written as $|OP(aip,t,k)|$, by scanning every IP pair in $IPpair(aip,t,k)$. In order to calculate $|OP(aip,t,k)|$, a key step is to acquire how many distinct IP pairs appear in $W(t,k)$. In another word, for a given IP pair $<aip,bip_j>$ which may appears several times in a slice, the problem is to determine if it appears in $W(t,k)$. This problem is simple in discrete time window where $k=1$ by using a single bit which will be set to 1 if $<aip,bip_j>$ appears. But in sliding time window when $k>1$, there are $k-1$ slices belong to two adjacent sliding time windows at the same time. For example, when $W(t,k)$ slides to $W(t+1,k)$, there are $k-1$ slices, slice $t+1$ to slice $t+k-1$, appearing in them at the same time. $<aip,bip_j>$ may appear in slice $t$ or in some slices after $t$. A sliding cardinality estimator must distinguish these cases and judge if $<aip,bip_j>$ appears in the new window after sliding. This paper devises a new recorder, distance recorder $DR$, to solve this problem. 

$DR$ is a recorder which consists of $z$ bits. It records the distance between the nearest slice where $<aip,bip_j>$ appears and the current scanning slice. For example, suppose that the estimator is now scanning $IPpair(aip,t,k)$ and $<aip,bip_j>$ appears in slice $t-d$ and not appears in slices after $t-d$. Then the value of $DR$ is $d$. When $d=0$, the distance will be 0 too which means $<aip,bip_j>$ appears in the current slice. Only when the distance is smaller than $k$ will host $bip_j$ appears in the sliding time window $W(t,k)$. So $z$ determines the max number of slices in a sliding time window and the max value of $k$ is $2^z-1$. In discrete time window, 1 bit is big enough for $DR$. When all of the $z$ bits in $DR$ is set to 1, it means that the distance is more than $k$ and $DR$ is also initialized to this value. $DR$ has four operations as listed below where $dr$, $dr_1$, $dr_2$ are instance of $DR$.
 \begin{description}
\item[DRinit($dr$)] set every bit of $dr$ to 1;
\item[DRset($dr$)] set every bit of $dr$ to 0;
\item[DRslide($dr$)] if the value of $dr$ is smaller than $2^z-1$, increment $dr$ by 1;
\item[DRjoin($dr_1$,$dr_2$)] return a new $DR$ which has the max value of $dr_1$ and $dr_2$
\end{description}

These operations make sure that $DR$ holds the correct distance for a certain IP pair or a certain host in network $B$. A precise way to calculate $|OP(aip,t,k)|$ is to allocate a $DR$ for every host in $OP(aip,t,k)$ and store these $DR$s by hash table or tree structures. This method could acquire the exact value of $|OP(aip,t,k)|$ by counting the number of $DR$ whose value are smaller than $k$ at the end of every slice. But it also consumes many memory and computing resource. For every host in $OP(aip,t,k)$, precise method requires $32+z$ bits, $32$ bits for IP address and $z$ bits for $DR$. The total memory requirement is more than $|OP(aip,t,k)|*(32+z)/8$ bytes. When $|OP(aip,t,k)|$ is very big, locating $DR$ of every host is also a hard task. So precise method is used to run offline to acquire baseline to evaluate the accuracy of other algorithms.

To saving memory and reducing processing time, estimator methods are required. When detection super point, an estimator only needs to tell if a host is super point or not. Under this requirement, a memory efficient algorithm, sliding rough estimator $SRE$, is devised.

For a host $aip$, the task of judging super point is to determine if $|OP(aip,t,k)| \geq \theta$ by scanning every host in $OP(aip,t,k)$ once. $SRE$ proposed in this paper is a memory efficient algorithm which can tell if a host is a super point in a time period with only $g$ $DR$s and 8 $DR$s are big enough for IPv4 address. Its weight $|SRE|^k$ is the number of $DR$ in it whose value is smaller than k. These $g$ $DR$s are initialized to $2^z-1$ at the begin of a time period. $SRE$ samples and records hosts in $IPpair(aip,t,k)$ by the least significant bits of their hashed value. Least significant bit of an integer is defined in the below.

\begin{definition}[Least significant bit, LSB]
\label{def-LeastSignifcantBit}
Given an integer $i$, let $BIN(i)$ represent its binary formatter. The least significant bit of $i$, $LSB(i)$, is the index of the first ‘1' bit of $BIN(i)$ starting from right.
\end{definition}

For example, $LSB(3)=0$, $LSB(40)=3$. The binary formatters of 3 and 40 are ``11" and ``101000". The first bit of $BIN(3)$ is 1, so $LSB(3)$ equals to 0. While $BIN(40)$ meets its first ‘1'until the fourth bit, so its $LSB$ is 3. For every host $bip$ in $ IPpair(aip,t,k)$, $SRE$ hashes it to a random value between 0 and $2^{32} -1$ by a hash function \cite{hash_AsmallApproximatelyMinWiseIndependentHF} $H_1$. If $LSB(H_1(bip))$ is smaller than an integer $\tau$, this IP will not be recorded by $SRE$ where $\tau$ is derived from $\theta$ by equation\ref{eq_getLsbThreshold}.

\begin{equation}\label{eq_getLsbThreshold}
\tau=ceil(log_2(\theta/ g ))
\end{equation} 

When $LSB(H_1(bip)) \geq \tau$, a bit selected by $H_2(bip)$ will be set where $H_2$ is another hash function mapping $bip$ to a value between 0 and $g -1$. After updating a $DR$, if $|SRE|^k$ is no smaller than $\rho * g$, $|OP(aip,t,k)|$ is judged as bigger than $\theta$ by $SRE$, where $\rho = 0.99 * (1-e^{-1/3})$. $\rho$ is acquired from \cite{DC:AnOptimalAlgorithmDistinctElementProblem}. $SRE$ deals with every host in $IPpair(aip,t,k)$ in this way. 

$SRE$ has a high probability to report a super point. Then we will give its mathematical analyze.

\begin{lemma}
\label{la-RE_fullN}
Suppose there are $\alpha$ different balls, $g$ different boxes and $\alpha \geq g$. Throw all of these balls randomly to these boxes. Let $FN(\alpha, g)$ represent the number of situations that every $g$ boxes has at least a ball. Then $FN(\alpha, g)= g ^ \alpha -\sum_{i=1}^{r-1}C_r^i*FN(\alpha, i)$ and $FN(\alpha,1) =1$.
\end{lemma} 
\begin{proof}
There are total $g ^ \alpha$ situations to threw $\alpha$ balls to $g$ boxes. When there is only a box, there is only a situation, throwing all balls to it. When throwing all balls to $i$ boxes and all of these boxes contain at least on balls, there are $C_r^i*FN(\alpha, i)$ situations. Deduct all situations that all balls are thrown to a subset of $g$ boxes from $g ^ \alpha$, the rest is the number of situations that there are no empty boxes.
\end{proof} 

\begin{theorem}
 \label{th-RE_NoneEmptyBoxesN}
Throw $\alpha$ balls to $g$ boxes. Let $g_1$ represent the number of boxes that contain at least a ball. The number of situations that there are $g_1$ balls are none empty,denoted by $FN(\alpha,g,g_1)$, is $C_{g}^{g_1}*FN(\alpha,n)$, where $1\leq n \leq g$.  
\end{theorem}
\begin{proof}
The rest $g-g_1$ balls are empty. There are $C_{g}^{n}$ situations to choose $g –g_1$ empty balls. Each situation has $FN(\alpha,g_1)$ methods to throw $\alpha$ balls. So the number of total situations is $C_{g}^{g_1}*FN(\alpha,g_1)$.
\end{proof}

$OP(aip,t,k)$ could be regarded as the set of balls and $g$ bits could be regarded as boxes in theorem \ref{th-RE_NoneEmptyBoxesN}. $|SRE|^k$ means the number of $DR$ whose values are smaller than $k$. Suppose there are $\alpha$ hosts in $OP(aip,t,k)$ updating $SRE$. The probability that there are $|SRE|^k = g_1$ is :

\begin{equation}\label{eqt_prb_alphaNoneEmpytBox}
 Pr\{\alpha,g,g_1\}=\frac{FN(\alpha,g,g_1)}{{g}^{\alpha}}
\end{equation}

Every host in $OP(aip,t,k)$ has probability $\frac{1}{2^\tau}$ to update $SRE$. So the probability that there are $\alpha$ hosts in $OP(aip,t,k)$ updating $SRE$ is:

\begin{equation}
\label{eqt_prb_alphaNum}
\begin{aligned}
 &Pr\{|OP(aip,t,k)|,\alpha\} \\
 &=C_{|OP(aip,t,k)|}^{\alpha}*{\frac{1}{2^ \tau}}^\alpha*(1-\frac{1}{2^\tau})^{|OP(aip,t,k)|-\alpha}
 \end{aligned}
\end{equation}

Combine equation \ref{eqt_prb_alphaNoneEmpytBox} and \ref{eqt_prb_alphaNum}, we will get the probability that there are $g_1$ $DR$ being set in $RE$ after scanning $ST(aip,t,k)$ as shown in equation \ref{eqt_prb_1BitsNAfterScanningHostStrm}.
\begin{equation}
\label{eqt_prb_1BitsNAfterScanningHostStrm} 
\begin{aligned}
&Pr\{|OP(aip,t,k)|,g,\tau,g_1\} \\
&=\sum_{\alpha=g_1}^{|OP(aip,t,k)|}Pr\{|OP(aip,t,k)|,\alpha\}*Pr\{\alpha,g,g_1\}
\end{aligned}
\end{equation}.

The probability that there are more than $n$ ‘1' bits in $SRE$ after scanning $ST(aip,t,k)$ could be derived from \ref{eqt_prb_1BitsNAfterScanningHostStrm} as shown in equation \ref{eqt_prb_1BitsNPlusAfterScanningHostStrm}.
\begin{equation}
\label{eqt_prb_1BitsNPlusAfterScanningHostStrm}
\begin{aligned} 
&Pr^+{|OP(aip,t,k)|,g,\tau,n} \\
&=\sum_{g_1=n}^{g}Pr\{|OP(aip,t,k)|,g,\tau,g_1\}
\end{aligned}
\end{equation}

Equation \ref{eqt_prb_1BitsNPlusAfterScanningHostStrm} proofs that $SRE$ has a high probability to detect super point. But it is a light weight estimator and can't give an accurate cardinality estimation. Sliding linear estimator introduced in the next makes up this shortage.
\subsection{Sliding linear estimator}
Linear estimator, $LE$, is a famous cardinality estimation algorithm\cite{DC:aLinearTimeProbabilisticCountingDatabaseApp}. It uses ${g}'$ bits, which are initialized to 0 at the beginning of a discrete time window, to estimate host's cardinality. When scanning a host $bip$ in $IPpair(aip,t,k)$, one bit in $LE$ selected by hash function $H_3$ will be set. $H_3(bip)$ maps $bip$ to a random value between 0 and ${g}'-1$. Let $|LE|$ represent the weight of $LE$, which means the number of 1 bit in it. At the end of a discrete time window, $|OP(aip,t,1)|$ will be estimated by the following equation.

\begin{equation}\label{eq_LE_cardinalityEst}
 |OP(aip,t,1)|= - {g}'*ln(\frac{{g}'-|LE|}{{g}'})
\end{equation} 

But $LE$ only works when $k=1$. In order to estimate cardinality under sliding time window, sliding linear estimator $SLE$ replaces the ${g}'$ bits in $LE$ with ${g}'$ $DR$s. The weight of $SLE$ denoted as $|SLE|^k$ is the number of short integer whose value is smaller than $k$. $SLE$ estimates a host's cardinality by equation \ref{eq_SLE_cardinalityEst}.

\begin{equation}
\label{eq_SLE_cardinalityEst}
 |OP(aip,t,k)|'= - {g}'*ln(\frac{{g}' - |SLE|^k}{{g}'})
\end{equation} 

According to paper \cite{DC:aLinearTimeProbabilisticCountingDatabaseApp}, the estimating accuracy of $SLE$ depends on the value of $g'$, the bigger $g'$ is, the more accurate the estimating result will be. But a big $g'$ requires more time to calculating $|SLE|^k$ which increasing the estimating time. So $SLE$ is only suit to estimating cardinality of candidate super points at the end of slice, instead of estimating every time while scanning IP pair. When combining $SRE$ with $SLE$, an novel sliding time window super point's cardinality estimating algorithm $SRLA$ is proposed.

\section{Detect super points and estimate their cardinalities on GPU}
Network $A$ contains a great number of hosts and it's not efficient to allocate a $SRE$ and $SLE$ for every host. This section introduces a novel algorithm which can detect super point and estimate their cardinalities under sliding time window with fixed number of estimators. 

\subsection{Scan packets and generate super points candidate list}
Because 8 $DR$s are big enough for $SRE$ to judge if a host is super point, it can detect super point fast. When using with $LE$, it can estimate super point's cardinality more quickly. Motivated by this idea, we design a novel estimator, sliding estimator $SE$. $SE$ consists of a $SRE$, a $LE$ and 16 bits. The 16 bits in $SE$ is used to indicate that a host has been judged as a super point in the time slice and we call them as super point indicator $SI$. When a host $aip \in A$ is firstly judged as a super point, a bit in $SI$, selected by a hash function $H_3(aip)$ where $H_3$ hashes $aip$ to a random value between [0,15], will be set to 1. Let $SI[i]$ point to the $i$th bit in $SI$. An array of $SE$ with $u$ rows and $v$ columns, denoted by $SEA$, is used to detect super points in the network $A$ and estimate their cardinalities. Figure \ref{fig_SEAmodel} illustrates the structure of $SEA$.
\begin{figure}[!ht]
\centering
\includegraphics[width=0.47\textwidth]{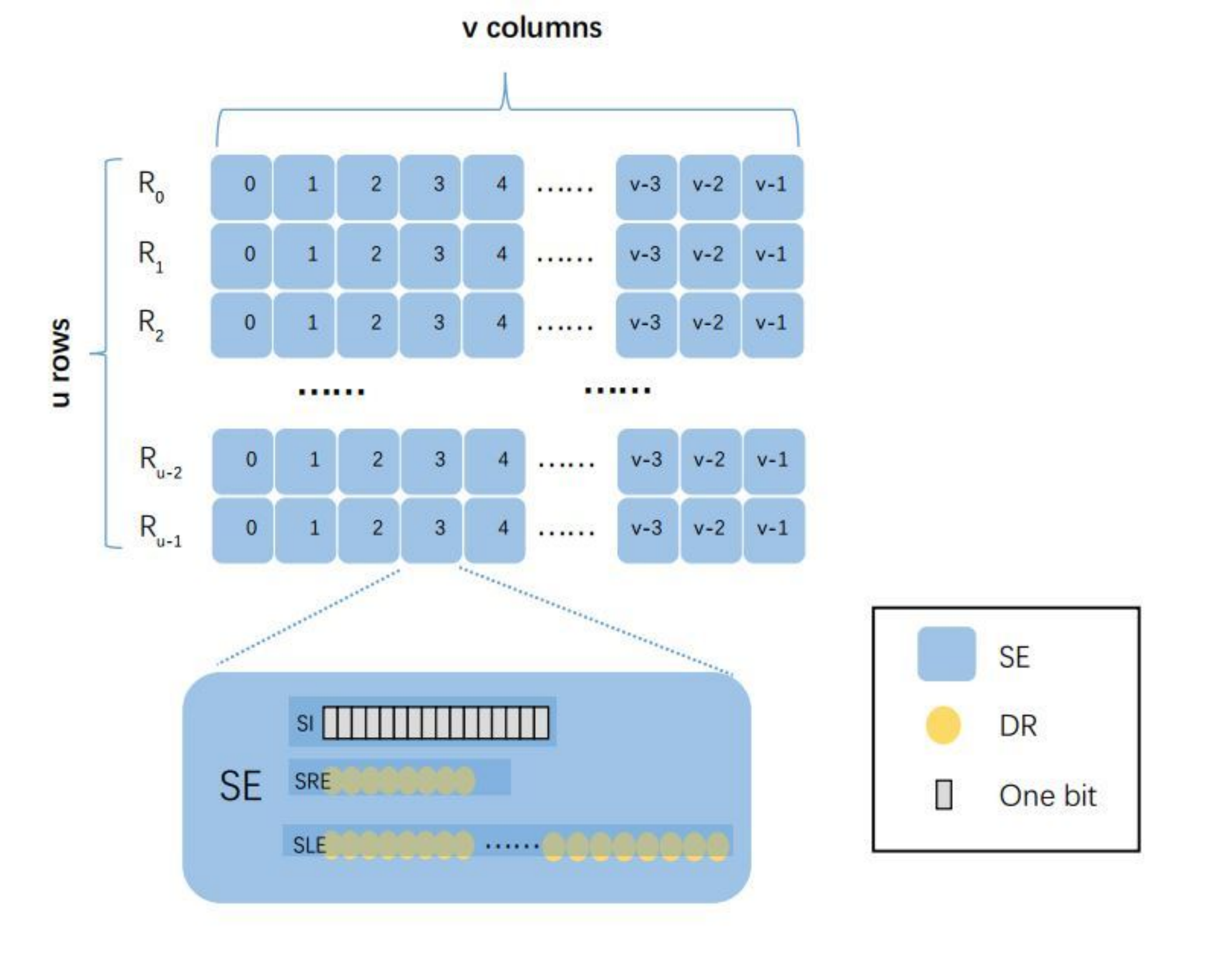}
\caption{Structure of SEA}
\label{fig_SEAmodel}
\end{figure} 

Every IP pair $<aip,bip>$ will update $u$ $SE$ selecting from the $u$ rows of $SEA$ by $u$ hash functions, $RH_i(aip)$ where $0 \leq i \leq u-1$. $USE(aip)$ represents the union $SE$ of these $u$ $SE$s in $SEA$ and $USI(aip)$, $URE(aip)$ and $ULE(aip)$ represent the $SI$, $SRE$ and $LE$ in $USE(aip)$ respectively. For a host $aip \in A$, its union $SE$ in $SEA$ is acquired by algorithm \ref{alg_aip_unionSE}. 

\begin{algorithm}  
\caption{UnionSE}  
\label{alg_aip_unionSE}
\begin{algorithmic}
\Require{ $aip \in A$;\\
         $SEA$}
\Ensure{ $USE(aip)$ the union of $SE$ relating with $aip$}

\State Init $USE(aip)$
\State set every bit of $SI$ in $USE(aip)$ to 1
\State set every $DR$ in $URE(aip)$ to $2^z-1$
\State set every $DR$ in $ULE(aip)$ to $2^z-1$
\For{$i \in [0, u-1]$}
\State $USI(aip) \Leftarrow USI(aip) \& SI[i,RH_i(aip)]$
\For{$j \in [0,g-1]$}
\State $URE(aip)[j] \Leftarrow DRjoin(URE(aip)[j], RE[i, RH_i(aip)][j])$
\EndFor
\For{$j \in [0,g'-1]$}
\State $ULE(aip)[j] \Leftarrow DRjoin(ULE(aip)[j], LE[i, RH_i(aip)][j])$
\EndFor
\EndFor
\State Return $USE(aip)$
\end{algorithmic}
\end{algorithm}

$SI[i,j]$, $RE[i,j]$ and $LE[i,j]$ is the $SI$, $SRE$ and $LE$ of $SE$ in the $i$th row, $j$th column.
After updating these $SE$s, $aip$ will be checked by $URE(aip)$ to test if it is a super point. If it is and $USI[H_4(aip)]$ is zero, $aip$ will be inserted in to a candidate super point list and the $H_4(aip)$th bit of every $SI[i,RH_i(aip)]$ will be set to 1. This will avoid to add $aip$ to the candidate super point list more times. 

To calculate $USE(aip)$ every time scanning an IP pair is time consuming, especially that the $g'$ is very big often more than one thousand. We only need to acquire $USI(aip)$ and $URE(aip)$ for super point judging and candidate super point list insertion. $SI$ contains only 16 bits and $SRE$ consists of only 8 $DR$ for IPv4 address. The merging time will be reduced greatly. Algorithm \ref{alg_scanIPpair_updateSEA} describes how to update $SEA$ for every IP pair.

\begin{algorithm}                       
\caption{ScanIPpair}          
\label{alg_scanIPpair_updateSEA}                           
\begin{algorithmic}
\Ensure{$SEA$,\\
       IP pair $<aip,bip>$ \\
       Candidate super point list $CSIP$
}
\State $leidx \Leftarrow H_1(bip)$
\For {$ridx \in [0,u-1]$}
 \State $LE[i,RH_i(aip)][leidx] \Leftarrow 0$
\EndFor

\If{$LSB(H_1(aip)) \leq \tau$}
\State Return
\EndIf
\State $reidx \Leftarrow H_2(bip)$
\State $siidx \Leftarrow H_4(bip)$
\For {$ridx \in [0,u-1]$}
  \State $RE[i,RH_i(aip)][reidx] \Leftarrow 0$
\EndFor
\If{$|URE(aip)|^k \geq \tau $}
\If{$USI(aip)[siidx]$ equal to 0}
  \State insert $aip$ into $CSIP$
   \For {$ridx \in [0,u-1]$}
    \State $SI[i,RH_i(aip)][siidx] \Leftarrow 1$  
\EndFor
\EndIf
\EndIf
\end{algorithmic}
\end{algorithm}

 Algorithm \ref{alg_scanIPpair_updateSEA} firstly updates $u$ $LE$s by setting a $DR$ of them to 0. Then it begins to update $SRE$ after checking $aip$. Not every IP pair could pass this checking and only about $\frac{1}{2^\tau}$ of them updates $SRE$. This checking process accelerates the scanning speed greatly. When a IP pair updates, the first IP address of it will be checked if is a super point by the union $SRE$. Super point reported by $SRE$ will be inserted into the candidate list $CSIP$. Algorithm \ref{alg_scanIPpair_updateSEA} deals with every IP pair in a slice. After scanning all IP pairs in this slice, the cardinality of hosts in the candidate super point list could be acquired from $SEA$ by the algorithm described in the next section.

\subsection{Estimate cardinality of super point}
$SEA$ uses fix number of $LE$, $u*v$ $SE$s, to estimate the cardinalities of all hosts in $A$. This causes that a $LE$ will record more than one hosts' cardinalities and the result will be over estimating. In order to reduce the influence, $u$ $LE$s will be used together and a host's cardinality will be estimated from the union $LE$. But when there are many distinct IP pairs in a slice, there are still many $DR$ in the union $LE$ setting by other hosts. Estimating the number of these error $DR$ and remove them from the union $LE$ helps to improve the accuracy of cardinality estimation. 

Let $|LDR(i)|^k$ represent the number of all $LE$s' $DR$ in the $i$th row whose values are smaller than $k$. Then the probability that a $DR$ of a $LE$ in the $i$th row is set by some host is $P_{dr}^{LE}(i)=\frac{|LDR(i)|^k}{g'*v}$. $|LDR(i)|^k$ could be acquired by scanning every $LE$ in the $i$th row. Suppose a $LE$ is used to record the cardinality of a host $aip$ exclusively. Then $|LE|^k$ is expected to be $d_1=g'-g'*e^{- \frac{|OP(aip,t,k)|}{g'}}$, according to equation \ref{eq_SLE_cardinalityEst}. In the union $LE$, every of these $d_1$ $DR$ will be set by some other hosts with probability $UP_{dr}^{LE}$ as shown in the following equation.
\begin{equation}
\label{equ_UnionLE_dr_setProb}
 UP_{dr}^{LE}=\prod_{i=0}^{u-1}P_{dr}^{LE}(i)
\end{equation}

Let $|ULE(aip)|^k$ represent the number of $DR$ in the union $LE$ whose values are smaller than $k$. Then $|ULE(aip)|^k= d_1+(g-d_1)*UP_{dr}^{LE}$. And $aip$'s cardinality could be estimated by the following equation.
\begin{equation}
\label{equ_estCardinality_ULE}
|OP(aip,t,k)|'=-g'*ln(1-\frac{|ULE|^k-g'*UP_{dr}^{LE}}{g'*(1-UP_{dr}^{LE})})
\end{equation}

Equation \ref{equ_estCardinality_ULE} gives a more accurate estimation by removing the error setting $DR$ from $ULE$. The cardinality of every host in the candidate super point list will be estimated in this way. 

$SRLA$ works under sliding time window. To do this, $SRLA$ updates $SEA$ incrementally instead of reinitialize it before every time slice. After estimating super point's cardinality, $SRLA$ updates all $SI$, $DR$ and the candidate super point list by algorithm \ref{alg_updateSEA_beforeNextTimeSlice}.

\begin{algorithm}                       
\caption{SEA updating before sliding}          
\label{alg_updateSEA_beforeNextTimeSlice}               
\begin{algorithmic}
\Require{$SEA$\\
        Candidate super point list $CSIP$}
\Ensure {New candidate super point list $NCSIP$}

\For{$si$ in $SI$ of every $SE$ in $SEA$}
\State $si \Leftarrow 0$
\EndFor
\For{$dr$ in $DR$ of all $SRE$ and $LE$ of $SE$ in $SEA$}
\If{$dr < 2^z-1$}
\State $dr++$
\EndIf
\EndFor

\For {$aip$ in $CSIP$}
 \If{$|URE(aip)|^k \geq g*\rho $} 
  \State insert $aip$ into $NCSIP$
   \For {$ridx \in [0,u-1]$}
    \State $SI[i,RH_i(aip)][siidx] \Leftarrow 1$  
\EndFor 
\EndIf
\EndFor

\State Return $NCSIP$
\end{algorithmic}
\end{algorithm}

Algorithm \ref{alg_updateSEA_beforeNextTimeSlice} not only updates all $DR$ in $SEA$ but also derives a new candidate super point list from now current one for the next time window. This makes sure that no super points will be neglected. For example, if $aip$ is a super point in $W(t+1,k)$ and all of its opposite hosts appear in time slice $t+1$ to $t+k-1$. In this case, $aip$ will not be inserted into the candidate super point list while scanning IP pairs in time slice $t+k$. But it could be detected out from the candidate super point lists in $W(t,k)$.

\subsection{Deploy on GPU}
While scanning IP pairs, $SRLA$ only sets some $SI$s, $DR$s. Bot $SI$s and $DR$ could be set by server threads at the same time without causing any mistakes, because a bit or a $DR$ is still being “1" or zero after setting several times. So several IP pairs could be processed concurrently. GPU is a special device which contains plenty computing cores and has high memory accessing through put. Although the ability of every single core of CPU is a little stronger than that of GPU, but the total computing resource of a GPU card is much more abundant than that of CPU considering the plenty number of cores a GPU containing. 

GPU is good at these tasks which process huge data with the same instructions. $SRLA$ is such one that scanning different IP pairs by algorithm \ref{alg_scanIPpair_updateSEA}. But GPU could only access its own memory directly, so these IP pairs should be stored in a buffer and then copied to GPU's graphic memory as shown in figure \ref{fig_GPU_scanIPpairModle}.

\begin{figure}[!ht]
\centering
\includegraphics[width=0.47\textwidth]{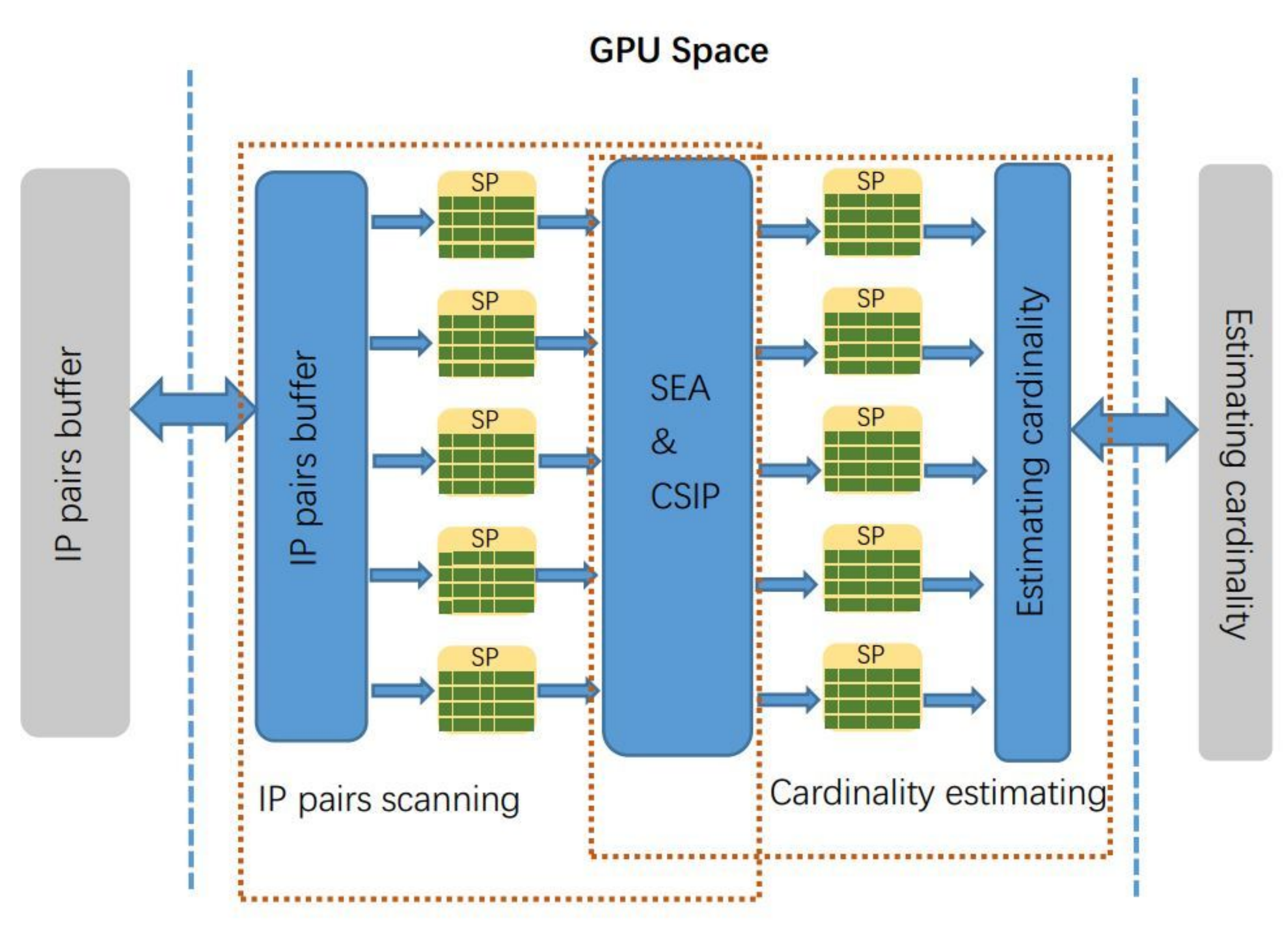}
\caption{Structure of SEA}
\label{fig_GPU_scanIPpairModle}
\end{figure}

Before $SRLA$ starting, $SEA$ will be initialized on GPU's graphic memory to be accessed by GPU threads directly. When the IP pairs buffer is full, it will be sent to GPU's global memory by PCIe bus. After receiving these IP pairs, GPU launches thousands of cores to deal with them at the same time.
Stream processor $SP$ is a set of hundreds of computing cores. A GPU card contains several $SP$s. Every $SP$ reads a part of IP pairs in the buffer and distributes them to different cores for further processing. Every core runs algorithm \ref{alg_scanIPpair_updateSEA} to update $SEA$ and candidate super point list $CSIP$ in a time slice.

After scanning all IP pairs in a slice, every computing core estimates cardinality of candidate super point by equation \ref{equ_estCardinality_ULE}.

Let $C_u$ represent the time of IP pairs scanning, $C_e$ represent the time of candidate super point's cardinality estimation and $C_s$ represent the duration of a slice. In order to deal with high speed network traffic in real time, $C_u+C_e$ must be smaller than $C_s$. Cardinality of candidate super point will not be estimated until the end of a slice, so the estimating latency under sliding time window is $C_s+C_e$. Experiments proves that for a 40Gb/s network, $SRLA$'s $C_e$ is as small as 300 milliseconds with a common GPU card and $C_u$ is smaller than 150 milliseconds. This shows that $SRLA$ works well under a sliding time window whose sliding step could be as small as 1 second when running on GPU.

\section{Experiment}
To evaluate the performance of $SRLA$, we use a real world traffic collecting from the node of JiangSu province of CERNET. The experiment data are two one-hour traffics starting from 13:00 on October 21 and 23, 2017. There are two parts in our experiments: super point cardinality estimation under discrete time window and super point cardinality under sliding time window. In both of these parts, super point's threshold $\theta$ is set to 1024. The experiment runs on a PC with GPU card Nvidia GTX 650, 1 GB graphic memory.

\subsection{discrete time window experiments}
The parameter of the discrete time window is set to $C_s=300$ seconds, $k=1$ and $z=1$. There are 12 discrete time windows in a one-hour traffic and the average information of these two traffics are listed in talbe \ref{tbl-trafficInf}.
\begin{table*}
\centering
\caption{Traffic information}
\label{tbl-trafficInf}
\begin{tabular}{c}                                                                                                                                                                                                                           
\centering
\includegraphics[width=0.8\textwidth]{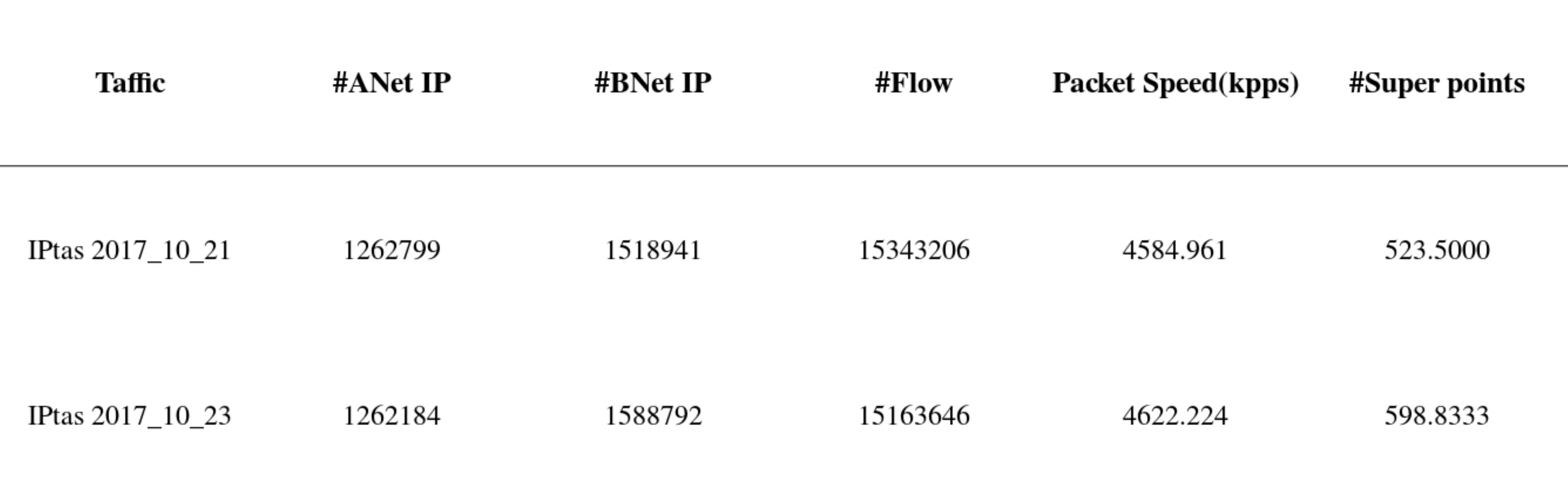}
\end{tabular}
\end{table*}

In table \ref{tbl-trafficInf}, “$\#ANet IP$" and “$\#BNet IP$" mean the number of hosts in $A$ and $B$ separately and “$\#Flow$" means the average number of distinct IP pairs in a discrete time window. From it we can see that, the average packets speed of this traffic is 4.5 mpps (million packets per second) and super point makes up smaller than 0.047 percent of the total hosts in $A$. 

Accuracy is a key merit of cardinality estimation. We measure the accuracy by false positive rate(FPR), false negative rate(FNR) as defined below.
\begin{definition}[FPR/FNR]
\label{def-fpr_fnr}
For a traffic with $N$ super points, an algorithm detects $N'$ super points. In the $N'$ detected super points, there are $N^+$ hosts which are not super points. And there are $N^-$ super points which are not detected by the algorithm. FPR means the ratio of $N^+$ to $N$ and FNR means the ratio of $N^-$ to $N$.
\end{definition}

FPR may decrease with the increase of FNR. If an algorithm reports more hosts as super point, its FNR will decrease but FPR will increase. So we use the sum of FPR and FNR, total false rate TFR, to evaluate the accuracy of an algorithm. 

The parameters of $SEA$ influences the accuracy of $SRLA$. We firstly compare the accuracy of $SRLA$ with different $u$, $v$ and $g'$. Figure \ref{SLRA_ComparingDifferentPara_IPtas2017_10_21} and \ref{SLRA_ComparingDifferentPara_IPtas2017_10_23} show the average accuracy of $SRLA$ under the 12 discrete time windows of different traffics and parameters.

\begin{figure*}[!ht]
\centering
\includegraphics[width=0.8\textwidth]{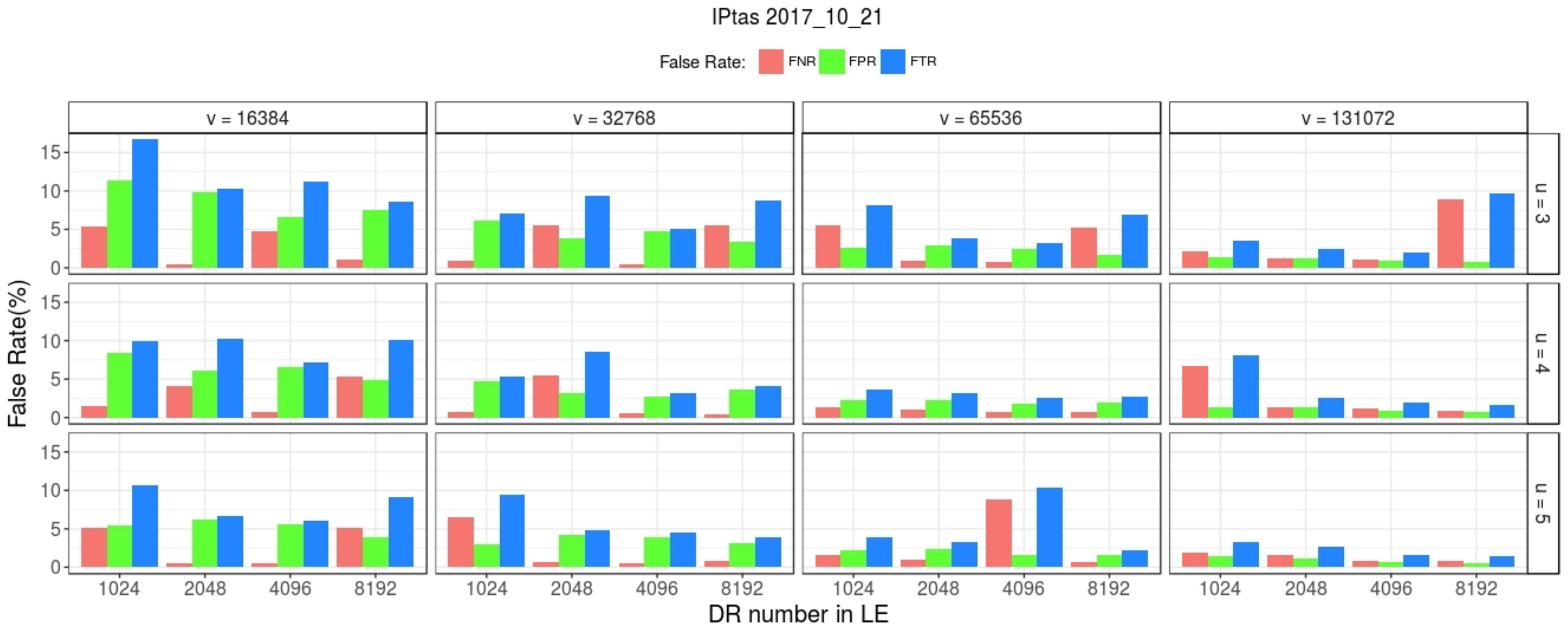}
\caption{SRLA accuracy comparing on traffic Oct, 21, 2017}
\label{SLRA_ComparingDifferentPara_IPtas2017_10_21}
\end{figure*}

\begin{figure*}[!ht]
\centering
\includegraphics[width=0.8\textwidth]{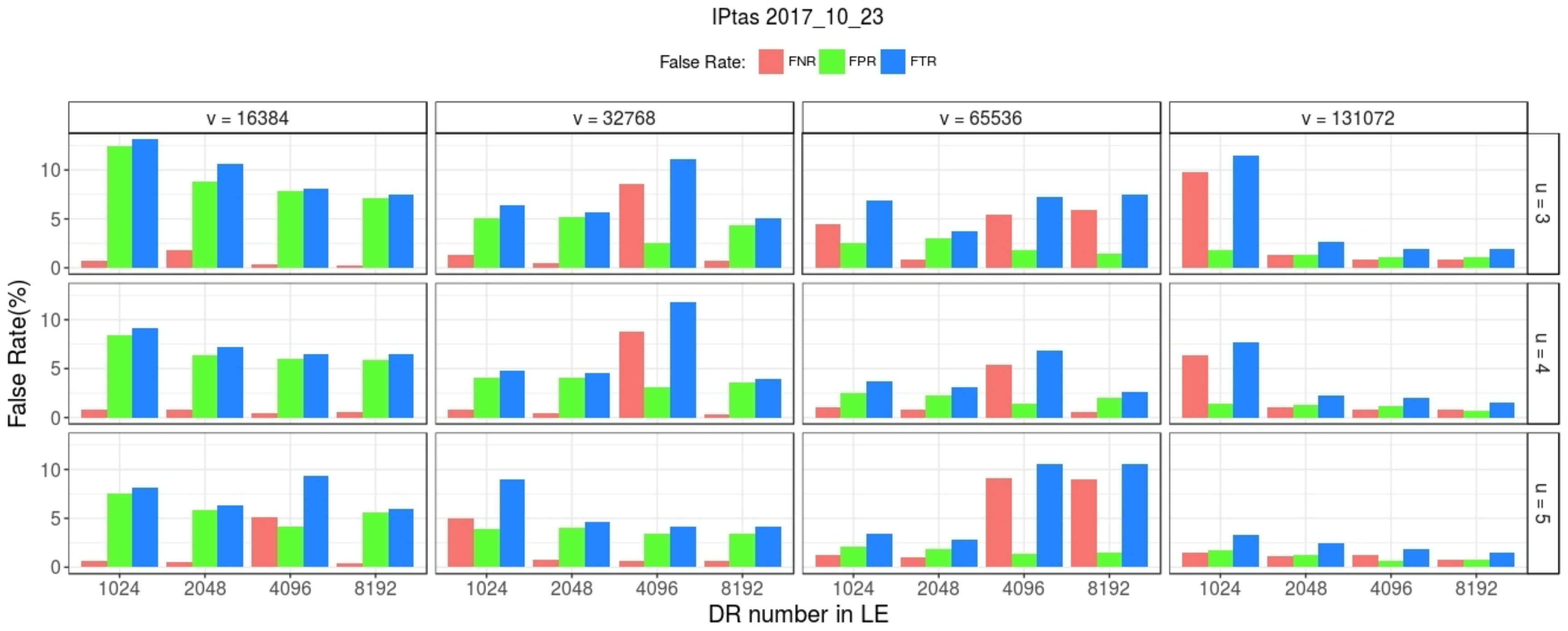}
\caption{SRLA accuracy comparing on traffic Oct, 23, 2017}
\label{SLRA_ComparingDifferentPara_IPtas2017_10_23}
\end{figure*}

Every sub figure compares the accuracy of $SRLA$ under different $g'$, changing from 1024 to 8192. Big $g'$ helps to reduce $TFR$ in most cases. But big $g'$ requires more time to acquire $|SLE|^k$ which will cause a big $C_e$. $TFR$ also decreases gradually with the increase of $v$. But when $v$ grows to 131072 from 65536, $TFR$ decreases slightly but memory doubles. $SRLA$ has the lowest false rate when $g'$, $u$ and $v$ are set to the biggest value. But memory requirement and $C_e$ also grow rapidly. When $g'=1024$, $v=65536$ and $u=4$, $SRLA$'s false rates are small enough and in the following experiments, $SRLA$'s parameters are set as these values. And when running under discrete time window, $z=1$ is enough for $DR$. 

To compare the performance of SRLA with other algorithms, we use DCDS\cite{HSD:ADataStreamingMethodMonitorHostConnectionDegreeHighSpeed}, VBFA\cite{HSD:DetectionSuperpointsVectorBloomFilter}, GSE \cite{HSD:GPU:2014:AGrandSpreadEstimatorUsingGPU} to compare with it. Table \ref{tbl_avg_hsd_rlt} lists the average result of all the 24 discrete time windows.

\begin{table*}
\centering
\caption{Comparing result under}
\label{tbl_avg_hsd_rlt}
\begin{tabular}{c}                                                                                                                                                                                                                           
\centering
\includegraphics[width=0.7\textwidth]{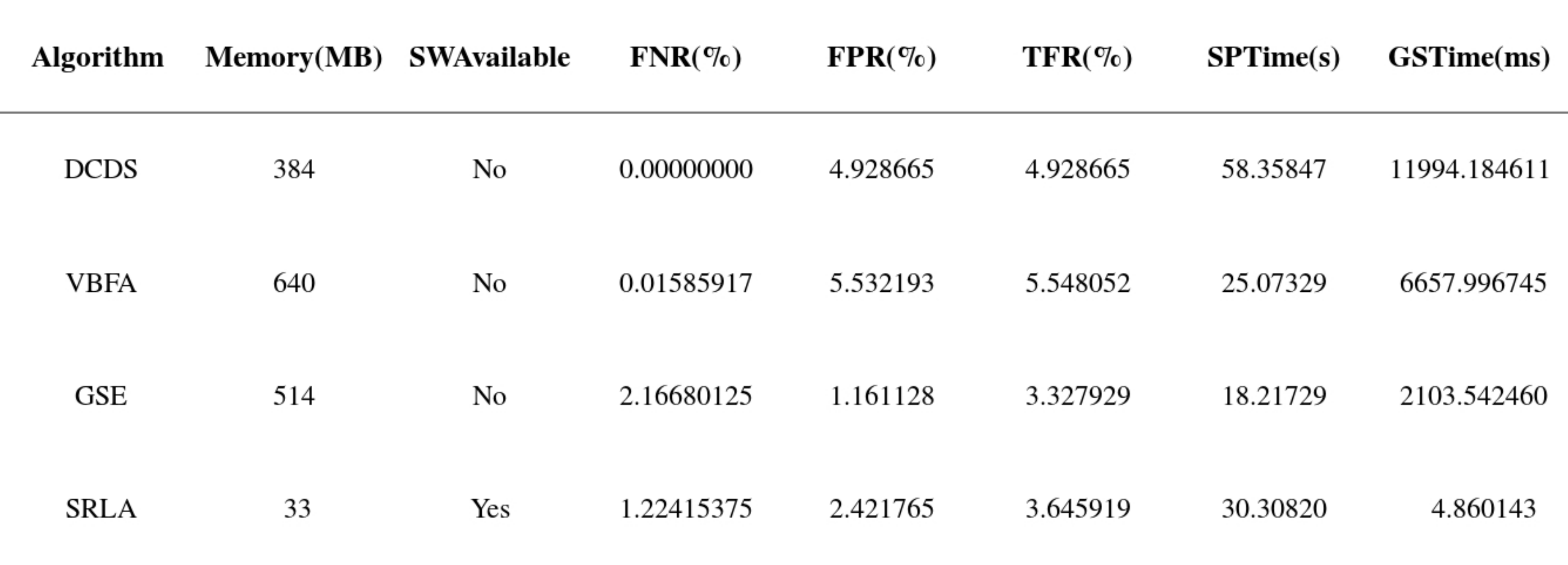}
\end{tabular}
\end{table*}

GSE has a lower FPR than other algorithms. It can remove fake super points according the estimating flow number. But GSE may remove some super points too, which causes it has a higher FNR. Because it uses discrete bits to record host's cardinality, collecting all of these bits together when estimate super points cardinality will use lots of time. DCDS uses CRT when storing host's cardinality. CRT has a better randomness which makes DCDS has a lower FNR. But CRT is very complex containing many operations. So DCDS's speed is the lowest among all of these algorithms. VBFA has the fastest speed but its TFR is higher than that of SRLA.

From table \ref{tbl_avg_hsd_rlt} we can see that, $SRLA$ uses the smallest memory, smaller than one-twentieth of others' memory. Because $SRLA$ generates a candidate super point list while packets scanning, so it has the smallest $C_e$, only 4 milliseconds. And $SRLA$ is the only one which can run under sliding time window.

\subsection{sliding time window experiments}

In the sliding time window experiments, a time slice is set to 1 second, $k$ is 300 and $z$ equals to 16. We let the window sliding from $W(0,300)$ to $W(2999,300)$ and $SRLA$ runs on traffic 2017-10-21. $SRLA$'s FPR, FNR and TFR are illustrated in figure \ref{SW_g16sv7_fpr_IPtas2017_10_21}, \ref{SW_g16sv7_fnr_IPtas2017_10_21} and \ref{SW_g16sv7_tfr_IPtas2017_10_21}.

\begin{figure*}[!ht]
\centering
\includegraphics[width=0.8\textwidth]{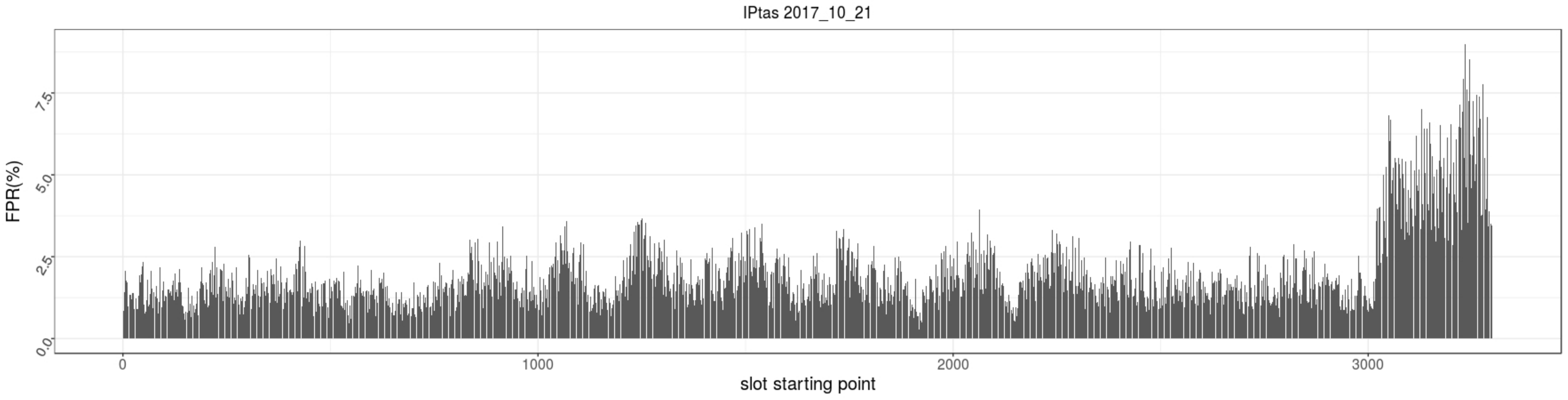}
\caption{FPR under sliding time window}
\label{SW_g16sv7_fpr_IPtas2017_10_21}
\end{figure*}

\begin{figure*}[!ht]
\centering
\includegraphics[width=0.8\textwidth]{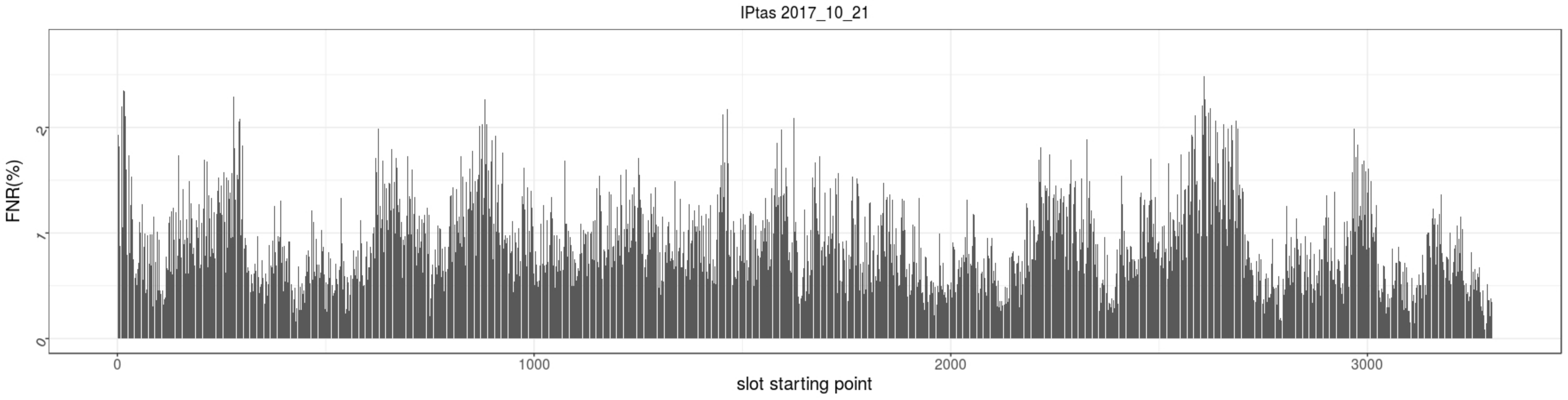}
\caption{FNR under sliding time window}
\label{SW_g16sv7_fnr_IPtas2017_10_21}
\end{figure*}

\begin{figure*}[!ht]
\centering
\includegraphics[width=0.8\textwidth]{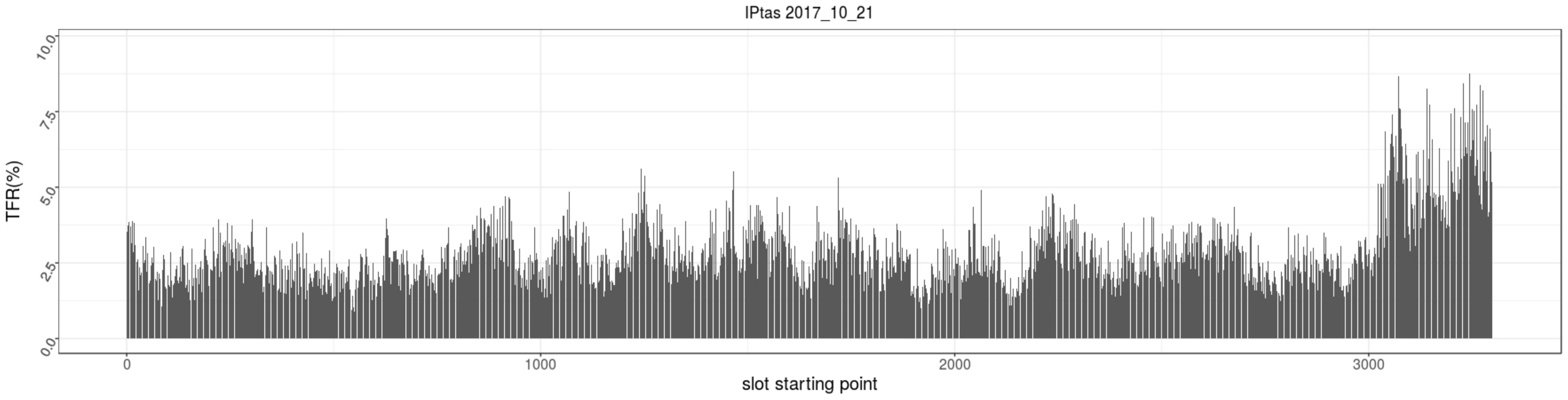}
\caption{TFR under sliding time window}
\label{SW_g16sv7_tfr_IPtas2017_10_21}
\end{figure*}

Under most sliding time window, $SRLA$ has a low FNR, smaller than $1.5\%$. When FNR is small, FPR is relative high. But the total false rate is stably small. When under sliding time window, $SRLA$ has the similar accuracy when it under discrete time window. This proves that $SRLA$ estimates super point cardinality successfully under sliding time window on GPU. In the sliding time window experiments, $SRLA$'s average $C_e$ is 100 milliseconds which is more than that under discrete time window. Because in sliding time window, $z$ is set to 16 and $SRLA$ requires more time to calculate $|SLE|^k$. But $C_e+C_u$ is still much smaller than $C_s$ and $SRLA$'s average $C_u$ is 109 milliseconds for every single slice. So $SRLA$ can detect and estimate the cardinality of super point in real time under sliding time window.

\section{Conclusion}
Super point cardinality estimation is an important and difficult task on network management. Incremental updating and small estimating time are two special difficulties in it. $SRLA$ proposed in this paper is the first one solve this problem in real time with a common GPU. $SRLA$'s capability of incremental updating comes from $DR$, a new recorder which can determine if itself is updated in a certain sliding time window. In order to reduce the super point's cardinality estimation time, $SRLA$ generating a candidate super point list while scanning IP pairs. This candidate super point list is acquired by the light weight sliding estimator $SRE$. $SRE$ is memory efficient and fast processing which makes sure that it doesn't cause many additional time for IP pairs scanning. At the end of a time slice, $SRLA$ estimates every candidate super point in the list by sliding linear estimator $SLE$. $SLE$ gives a high accuracy estimation of a host's cardinality. When running on a common GPU, $SRLA$ estimates the super point cardinality in real time for a 40Gb/s network.

\section{Reference}

 \bibliographystyle{elsarticle-num}

\bibliography{..//..//ref} 

\begin{thebibliography}{10}
\providecommand{\url}[1]{#1}
\csname url@samestyle\endcsname
\providecommand{\newblock}{\relax}
\providecommand{\bibinfo}[2]{#2}
\providecommand{\BIBentrySTDinterwordspacing}{\spaceskip=0pt\relax}
\providecommand{\BIBentryALTinterwordstretchfactor}{4}
\providecommand{\BIBentryALTinterwordspacing}{\spaceskip=\fontdimen2\font plus
\BIBentryALTinterwordstretchfactor\fontdimen3\font minus
  \fontdimen4\font\relax}
\providecommand{\BIBforeignlanguage}[2]{{%
\expandafter\ifx\csname l@#1\endcsname\relax
\typeout{** WARNING: IEEEtran.bst: No hyphenation pattern has been}%
\typeout{** loaded for the language `#1'. Using the pattern for}%
\typeout{** the default language instead.}%
\else
\language=\csname l@#1\endcsname
\fi
#2}}
\providecommand{\BIBdecl}{\relax}
\BIBdecl

\bibitem{JNCA2015_BotFlexACommunityDrivenToolForBotnetDetection}
\BIBentryALTinterwordspacing
S.~Khattak, Z.~Ahmed, A.~A. Syed, and S.~A. Khayam, ``Botflex: A
  community-driven tool for botnet detection,'' \emph{Journal of Network and
  Computer Applications}, vol.~58, pp. 144 -- 154, 2015. [Online]. Available:
  \url{http://www.sciencedirect.com/science/article/pii/S1084804515002155}
\BIBentrySTDinterwordspacing

\bibitem{JNCA2016:BotnetDetectionViaMiningTrafficFlowCharacteristics}
\BIBentryALTinterwordspacing
G.~Kirubavathi and R.~Anitha, ``Botnet detection via mining of traffic flow
  characteristics,'' \emph{Computers Electrical Engineering}, vol.~50, pp. 91
  -- 101, 2016. [Online]. Available:
  \url{http://www.sciencedirect.com/science/article/pii/S0045790616000148}
\BIBentrySTDinterwordspacing

\bibitem{JNCA2017:AValidationModelForNonlexicalRoutingProtocols}
\BIBentryALTinterwordspacing
H.~Khayou and B.~Sarakbi, ``A validation model for non-lexical routing
  protocols,'' \emph{Journal of Network and Computer Applications}, vol.~98,
  pp. 58 -- 64, 2017. [Online]. Available:
  \url{http://www.sciencedirect.com/science/article/pii/S1084804517303028}
\BIBentrySTDinterwordspacing

\bibitem{HSD:AbitmapBasedAlgorithmDetectingStealthySuperpoints}
Z.~Li, W.~Liu, Z.~Li, and J.~Sun, ``A bitmap-based algorithm for detecting
  stealthy superpoints,'' in \emph{2014 IEEE 12th International Conference on
  Dependable, Autonomic and Secure Computing}, Aug 2014, pp. 333--337.

\bibitem{HSD:2014ANewSketchMethodMeasuringHostConnectionDegreeDistribution}
P.~Wang, X.~Guan, J.~Zhao, J.~Tao, and T.~Qin, ``A new sketch method for
  measuring host connection degree distribution,'' \emph{IEEE Transactions on
  Information Forensics and Security}, vol.~9, no.~6, pp. 948--960, June 2014.

\bibitem{HSD:identifyHighCardinalityHosts}
J.~Cao, Y.~Jin, A.~Chen, T.~Bu, and Z.~L. Zhang, ``Identifying high cardinality
  internet hosts,'' in \emph{IEEE INFOCOM 2009}, April 2009, pp. 810--818.

\bibitem{HSD:streamingAlgorithmFastDetectionSuperspreaders}
S.~Venkataraman, D.~Song, P.~B. Gibbons, and A.~Blum, ``New streaming
  algorithms for fast detection of superspreaders,'' in \emph{in Proceedings of
  Network and Distributed System Security Symposium (NDSS}, 2005, pp. 149--166.

\bibitem{HSD:findFrequentItemsInStream}
\BIBentryALTinterwordspacing
G.~Cormode and M.~Hadjieleftheriou, ``Finding the frequent items in streams of
  data,'' \emph{Commun. ACM}, vol.~52, no.~10, pp. 97--105, Oct. 2009.
  [Online]. Available: \url{http://doi.acm.org/10.1145/1562764.1562789}
\BIBentrySTDinterwordspacing

\bibitem{JNCA2018:MATEMAUnifiedFrameworkTustMCDMAssuringSecurityReliabilityQosDTNRouting}
\BIBentryALTinterwordspacing
A.~B. Paul, S.~Biswas, S.~Nandi, and S.~Chakraborty, ``Matem: A unified
  framework based on trust and mcdm for assuring security, reliability and qos
  in dtn routing,'' \emph{Journal of Network and Computer Applications}, vol.
  104, pp. 1 -- 20, 2018. [Online]. Available:
  \url{http://www.sciencedirect.com/science/article/pii/S1084804517304010}
\BIBentrySTDinterwordspacing

\bibitem{JNCA2018:ZeroqueueEthernetCongestionControlProtocolBasedOnAvailableBandwidthEstimation}
\BIBentryALTinterwordspacing
M.~Bahnasy, H.~Elbiaze, and B.~Boughzala, ``Zero-queue ethernet congestion
  control protocol based on available bandwidth estimation,'' \emph{Journal of
  Network and Computer Applications}, vol. 105, pp. 1 -- 20, 2018. [Online].
  Available:
  \url{http://www.sciencedirect.com/science/article/pii/S108480451730423X}
\BIBentrySTDinterwordspacing

\bibitem{JNCA2018:SmartResourceAllocationForImprovingQoEIPMultimediaSubsystems}
\BIBentryALTinterwordspacing
A.~C谩novas, M.~Taha, J.~Lloret, and J.~Tom谩s, ``Smart resource allocation
  for improving qoe in ip multimedia subsystems,'' \emph{Journal of Network and
  Computer Applications}, vol. 104, pp. 107 -- 116, 2018. [Online]. Available:
  \url{http://www.sciencedirect.com/science/article/pii/S1084804517304277}
\BIBentrySTDinterwordspacing

\bibitem{JNCA2016:SoftwareDefinedNetworksASurvey}
\BIBentryALTinterwordspacing
R.~Masoudi and A.~Ghaffari, ``Software defined networks: A survey,''
  \emph{Journal of Network and Computer Applications}, vol.~67, pp. 1 -- 25,
  2016. [Online]. Available:
  \url{http://www.sciencedirect.com/science/article/pii/S1084804516300297}
\BIBentrySTDinterwordspacing

\bibitem{HSD:bitmapCountingActiveFlowsHighSpeedLinks}
\BIBentryALTinterwordspacing
C.~Estan, G.~Varghese, and M.~Fisk, ``Bitmap algorithms for counting active
  flows on high-speed links,'' \emph{IEEE/ACM Trans. Netw.}, vol.~14, no.~5,
  pp. 925--937, Oct. 2006. [Online]. Available:
  \url{http://dx.doi.org/10.1109/TNET.2006.882836}
\BIBentrySTDinterwordspacing

\bibitem{HSD:ADataStreamingMethodMonitorHostConnectionDegreeHighSpeed}
P.~Wang, X.~Guan, T.~Qin, and Q.~Huang, ``A data streaming method for
  monitoring host connection degrees of high-speed links,'' \emph{IEEE
  Transactions on Information Forensics and Security}, vol.~6, no.~3, pp.
  1086--1098, Sept 2011.

\bibitem{HSD:DetectionSuperpointsVectorBloomFilter}
W.~Liu, W.~Qu, J.~Gong, and K.~Li, ``Detection of superpoints using a vector
  bloom filter,'' \emph{IEEE Transactions on Information Forensics and
  Security}, vol.~11, no.~3, pp. 514--527, March 2016.

\bibitem{HSD:GPU:2014:AGrandSpreadEstimatorUsingGPU}
\BIBentryALTinterwordspacing
S.-H. Shin, E.-J. Im, and M.~Yoon, ``A grand spread estimator using a graphics
  processing unit,'' \emph{Journal of Parallel and Distributed Computing},
  vol.~74, no.~2, pp. 2039 -- 2047, 2014. [Online]. Available:
  \url{http://www.sciencedirect.com/science/article/pii/S0743731513002189}
\BIBentrySTDinterwordspacing

\bibitem{hash_AsmallApproximatelyMinWiseIndependentHF}
\BIBentryALTinterwordspacing
P.~Indyk, ``A small approximately min-wise independent family of hash
  functions,'' \emph{Journal of Algorithms}, vol.~38, no.~1, pp. 84 -- 90,
  2001. [Online]. Available:
  \url{http://www.sciencedirect.com/science/article/pii/S0196677400911313}
\BIBentrySTDinterwordspacing

\bibitem{DC:AnOptimalAlgorithmDistinctElementProblem}
\BIBentryALTinterwordspacing
D.~M. Kane, J.~Nelson, and D.~P. Woodruff, ``An optimal algorithm for the
  distinct elements problem,'' in \emph{Proceedings of the Twenty-ninth ACM
  SIGMOD-SIGACT-SIGART Symposium on Principles of Database Systems}, ser. PODS
  '10.\hskip 1em plus 0.5em minus 0.4em\relax New York, NY, USA: ACM, 2010, pp.
  41--52. [Online]. Available: \url{http://doi.acm.org/10.1145/1807085.1807094}
\BIBentrySTDinterwordspacing

\bibitem{DC:aLinearTimeProbabilisticCountingDatabaseApp}
\BIBentryALTinterwordspacing
K.-Y. Whang, B.~T. Vander-Zanden, and H.~M. Taylor, ``A linear-time
  probabilistic counting algorithm for database applications,'' \emph{ACM
  Trans. Database Syst.}, vol.~15, no.~2, pp. 208--229, Jun. 1990. [Online].
  Available: \url{http://doi.acm.org/10.1145/78922.78925}
\BIBentrySTDinterwordspacing

\bibitem{expdata:IPtraceCernetJS}
CERNET, ``China education and research network,''
  \url{http://iptas.edu.cn/src/system.php}, 2017, online;accessed 2017.

\end{thebibliography}

\end{document}